\documentclass[10pt]{article}

\usepackage[top=2cm, bottom=2cm, left=2cm, right=2cm]{geometry}
\usepackage{amsmath,amssymb,amsfonts,amsthm,MnSymbol}
\usepackage{graphics,subfigure,color} 
\usepackage{latexsym}
\usepackage[dvips]{graphicx}
\usepackage{epsfig}
\usepackage{hyperref}
\usepackage{chngcntr}
\usepackage{psfrag}
\usepackage{multirow}
\usepackage{braket}
\usepackage{float}
\usepackage{enumerate}

\newcommand{\be}{\begin{equation}}
\newcommand{\ee}{\end{equation}}

\newcommand{\cC}{{\cal{C}}}
\newcommand{\cD}{{\cal{D}}}

\newcommand{\cF}{{\cal{F}}}
\newcommand{\cG}{{\cal{G}}}
\newcommand{\cH}{{\cal{H}}}

\newcommand{\cU}{{\cal U}}

\newcommand{\mbT}{\mathbb{T}}

\DeclareMathOperator{\Tr}{Tr}

\newtheorem{theorem}{Theorem}

\theoremstyle{definition}

\allowdisplaybreaks[4]

\begin{document}

\title{Symmetry breaking in tensor models}
 
\author{
\ Dario Benedetti\footnote{dario.benedetti@th.u-psud.fr, Laboratoire de Physique Th\'eorique, CNRS-UMR 8627, Universit\'e Paris-Sud 11, 91405 Orsay Cedex, France},
 \ Razvan Gurau\footnote{rgurau@cpht.polytechnique.fr, Centre de Physique Th\'eorique, CNRS UMR 7644, \'Ecole Polytechnique, 91128 Palaiseau Cedex, France
 and Perimeter Institute for Theoretical Physics, 31 Caroline St. N, N2L 2Y5, Waterloo, ON, Canada.}
}

\maketitle

\begin{abstract} In this paper we analyze a quartic tensor model with one interaction for a tensor of arbitrary rank. This model 
has a critical point where a continuous limit of infinitely refined random geometries is reached.
We show that the critical point corresponds to a phase transition in the tensor model associated to a breaking of the unitary symmetry. 
We analyze the model in the two phases and prove that, in a double scaling limit, the symmetric phase corresponds to a theory of infinitely refined 
random surfaces, while the broken phase corresponds to a theory of infinitely refined random nodal surfaces. At leading order in the double scaling limit
planar surfaces dominate in the symmetric phase, and planar nodal surfaces dominate in the broken phase.
\end{abstract}

\section{Introduction}

Tensor models \cite{ambj3dqg,sasa1,color,review} generalize matrix models \cite{DiFrancesco:1993nw} and generate 
dynamical triangulations \cite{EDTDavid,EDTAmbjorn} in any dimension. The number of simplices in a triangulation 
diverges when the coupling constants are tunned to some critical values \cite{Kazakov:1985ds,David:1984tx,critical,uncoloring}. 
Sending at the same time the edge length of the individual simplices to zero, so as to keep the physical dimensionfull volume of the triangulation fixed, matrix and tensor models
describe a theory of random infinitely refined spaces \cite{DiFrancesco:1993nw,EDTAmbjorn}.

Like matrix models \cite{'tHooft:1973jz}, tensor models are endowed with a small parameter, $1/N$ (where $N$ is the size of the tensor) and exhibit a $1/N$ expansion
\cite{expansion1,expansion2,expansion3,uncoloring,review}. However, while in matrix models each order in $1/N$ corresponds to a fixed topology (i.e. to fixed genus surfaces) and the $1/N$ expansion is topological,
the same can not hold true in higher dimensions. Indeed, for $D\ge 3$, there does not exist a single topological invariant that distinguishes topologies which could be used to 
index a topological expansion. Consequently the $1/N$ expansion is not topological in dimension higher than two. However, it is as close as one can get to one:
the different topologies are explored progressively (i.e. at any order in $1/N$ only a finite number of topologies contribute) 
and at leading order only the spherical topology contributes.

The triangulated spaces generated by a tensor model in rank $D$ have topological dimension $D$. However, the emergent continuous geometry in the critical regime has typically a different effective dimension. While in the case of matrices the emergent geometry at leading order is a Brownian sphere \cite{BrownMap1,BrownMap2,BrownMap3}, with Hausdorff dimension 4 and spectral dimension 2, 
for tensor models the emergent geometry at leading order is typically a branched polymer \cite{melbp} (a continuous random tree in the mathematical literature), with Hausdorff dimension 2 and spectral dimension 4/3.
In order to uncover a genuine higher dimensional random geometry using tensor models one needs to go beyond the simple critical behavior at leading order.
A recent proposition consists in boosting some subleading contributions \cite{Bonzom:2015axa} in tensor models. With these non canonical scalings one can write tensor models with 
a richer phase structure having notably a Brownian sphere phase. 

We follow here an alternative route towards a non branched polymer phase in higher dimension. Our approach is inspired by the 
double scaling limit \cite{double,double1,double2,DGR,GurSch,Bonzom:2014oua}. In the critical regime, the contributions of the lower orders in $1/N$ are enhanced. 
Taking at the same time the coupling constants to criticality and $N$ to infinity
one obtains a continuous phase in which many orders in $1/N$ contribute. In this way one can hope that 
the contributions of this larger family of graphs pile up and eventually lead to a new geometric phase. 
 
The double scaling limit of tensor models ($D\ge 3$) has so far been studied 
for the case of actions which are invariant under a global permutation of the indices of the tensor
(i.e. for which all the indices of the tensor are treated on an equal footing). 
In this case the leading family in double scaling turns out to be exponentially bounded for $3\le D<6$ and
not exponentially bounded for $D>6$  \cite{DGR} (for matrices $D=2$, the leading family in  double scaling is not exponentially bounded).
 While not proven, it is very likely that all the 
leading order graphs in double scaling for $3\le D<6$ have the topology of a $D$ dimensional sphere. However, as this family proliferates like
a family of trees (the \emph{cherry trees} \cite{DGR}), it is very likely that for models possessing a color permutation symmetry the double scaling limit explored so far
leads again to a branched polymer phase. 

In this paper we study in detail a tensor model with only one quartic interaction.
The main difference with respect to the models studies previously is that in our case the indices of the tensor are \emph{not} treated on an equal footing: 
as we will see below the action is \emph{not} invariant under a global permutation of the indices of the tensor.

We study the critical behavior of this model and show that (similarly to the color symmetric quartic model \cite{Delepouve:2015nia})
the critical point of the dynamical triangulation continuum limit corresponds to a phase transition in the tensor model. 
However, the similarities between the model with only one interaction and the color symmetric model end here. As the 
non symmetric model is simpler than the symmetric one, we are able to solve the equations of motion in both phases and identify 
the physical vacua of the theory on the two sides of the phase transition. We are then able to analyze the approach to criticality (in double scaling)
on both sides of this phase transition.

The main result of this paper is that, for the model with only one interaction,
one obtains a genuinely new geometric phase in double scaling: instead of obtaining (at leading order in double scaling) a family of trees, 
one obtains (at leading order in double scaling) a family of planar graphs. The phase transition (in double scaling) is a phase transition 
between a phase of planar surfaces and a phase of planar nodal surfaces.

Of course the emergent geometry obtained in double scaling is not yet genuinely $D$ dimensional, but (very likely) is a Brownian sphere.
However, this is a first step in the quest for a genuinely new random geometry in higher dimensions: 
we succeeded for the first time in finding a non branched polymer state of higher dimensional random geometry in a scaling regime of tensor models. 

Besides their importance for the invariant tensor models, our results could be relevant for group and tensor 
field theories \cite{Oriti:2011jm,Baratin:2011aa,Krajewski:2012aw,tt2,Rivasseau:2013uca,BenGeloun:2011rc,Geloun:2014kpa,Samary:2014oya,Lahoche:2015ola}.
A regime parallel to the double scaling limit we discuss in this paper could be reached naturally by such models 
under the renormalization group flow  \cite{BenGeloun:2012yk,Carrozza:2014rba}.

This paper is organized as follows. We introduce the quartic tensor model with one interaction in section 
\ref{sec:model}. We obtain the vacua of this model and show that the critical point corresponds to a phase 
transition in section \ref{sec:vacua}.  We discuss the effective theory around a vacuum state in sections \ref{sec:intout} and \ref{sec:efft}. Finally 
we analyze the random geometry underlying the effective theory in the two phases of the model in section \ref{sec:phgeom}.

\section{The model}\label{sec:model}

An introduction to tensor models in general and to the quartic models in particular can be found in \cite{uncoloring,review,Nguyen:2014mga,expansioin6}.
Let us consider a rank $D$ tensor ${\mbT}_{a^1 \dots a^D}$, where $a^c= 1\dots N$ (and let use denote $\bar{\mbT}_{a^1\dots a^D}$ its complex conjugated) 
transforming under the external tensor product of 
$D$ fundamental representations of the unitary group ${\cal U}(N)$:
\begin{equation}\label{eq:transf}
  \mbT'_{a^{\cD}} = \sum_{b^1,\dots b^D=1}^N U^{(1)}_{a^1b^1} \dots U^{(D)}_{a^Db^D}  \mbT_{b^{\cD}} \;, \qquad
 \bar \mbT'_{  a^{\cD}} = \sum_{ b^1,\dots   b^D=1}^N \bar U^{(1)}_{  a^1   b^1} \dots \bar U^{(D)}_{  a^D   b^D}  \bar \mbT_{ b^{\cD}} \;.
\end{equation}
The position $c\in \cD = \{1,\dots D\}$ of an index is called its \emph{color}. 
A tensor model is a probability distribution for $\mbT$ and $\bar \mbT$ invariant under the transformations
in Eq.~\eqref{eq:transf}. The partition function of the quartic melonic tensor model symmetric under permutations of the colors is:
\begin{align}\label{eq:4model}
Z^{\rm{sym.}}(g) &= \int \left(  \prod_{n^{\cD}}  N^{D-1} \frac{ d \bar \mbT_{n^{\cD}}  d\mbT_{n^{\cD}} }{2  i  \pi} \right) \; 
e^{ -N^{D-1} \left(  \sum_{n^{\cD} \bar n^{\cD}} \bar \mbT_{ \bar n^{\cD} }  \delta_{ \bar n^{\cD} n^{\cD}} 
  \mbT_{ n^{\cD}}   - \frac{g^2}{2} \sum_{ c \in \cD } V^{(4)}_c ( \bar \mbT,  \mbT  )    
   \right) }  \;, \crcr
    V^{(4)}_c( \bar \mbT,  \mbT  )    &= \sum_{\bar n^{\cD} n^{\cD} m^{\cD} \bar m^{\cD}}
        \left( \bar \mbT_{\bar n^{\cD}} \delta_{\bar n^{\cD \setminus \{c\}} n^{\cD\setminus \{c\}}} \mbT_{n^{\cD}} \right)    \delta_{\bar n^{c} m^{c} } \delta_{\bar m^{c} n^{c}}
        \left( \bar \mbT_{\bar m^{\cD}} \delta_{\bar m^{\cD \setminus \{c\}} m^{ \cD \setminus \{c\}} }\mbT_{m^{\cD}} \right) \; ,
\end{align}
where, for any set $\cC \subset \cD$, we denote $n^\cC = (n^c, c\in\cC)$ and $ \delta_{\bar n^{\cC} n^{\cC}} = \prod_{c\in \cC} \delta_{\bar n^c n^c}$. 
Each invariant $ V^{(4)}_c( \bar \mbT,  \mbT  )  $ represents a gluing of four $D$-simplices, two positively and two negatively oriented. The coupling 
$g$ counts the number of positively oriented simplices. When evaluating this partition function in perturbation theory one obtains Feynman graphs having
$(D+1)$ colors which are dual to $D$-dimensional triangulations \cite{review}.
  
This model can be studied in the intermediate field representation \cite{Nguyen:2014mga,expansioin6}.
Each quartic interaction can be represented as Gaussian integral over an auxiliary $N \times N$ hermitian matrix:
\be
 e^{ N^{D-1} \frac{g^2}{2} V^{(4)}_c ( \bar \mbT,  \mbT  ) }  
 =\int [dH_c]\;   e^{-\frac{1}{2} N^{D-1} \Tr[H^cH^c] +g N^{D-1}\sum_{n^c \bar n^c} H^c_{\bar n^c n^c} \sum_{n^{\cD \setminus \{ c \} } \bar n^{\cD \setminus \{ c \} }} 
 \left( \bar \mbT_{\bar n^{\cD}} \delta_{\bar n^{\cD \setminus \{c\}} n^{\cD\setminus \{c\}}} \mbT_{n^{\cD}} \right) 
 } \;, 
\ee
normalized to $1$ for $g=0$. As there are $D$ quartic melonic interactions (one for each $c\in \cD$), 
one must introduce $D$ intermediate matrix fields $H^c$.  
The original tensor $\mbT$ and $\bar \mbT$ can be integrated out to obtain: 
\begin{align}
  Z^{\rm{sym.}}(g)   =  \int \left( \prod_{c=1}^D [dH^{c}] \right) \; 
    e^{ - S(H)  } \; , \qquad S(H) = \frac{1}{2} N^{D-1}  \sum_{c=1}^{ D}   \Tr_{c}[H^{c}H^{c}]  -  \Tr_{\cD} \left[ \ln R(H) \right]  \;,
\end{align}
where $\Tr_c$ is the trace over the index of color $c$, $\Tr_{\cD}$ is the trace over all the indices, and:
\be
 R(H)  =  \frac{1}{  \mathbb{I}^{\cD}- g \sum_{ c=1}^D  H^{c} \otimes \mathbb{I}^{\cD\setminus \{c\}}   } \; ,
\ee
where we denoted $\mathbb{I}^c$ the identity matrix with indices of color $c$ and $\mathbb{I}^{\cC} = \bigotimes_{c\in \cC} \mathbb{I}^{c}$. 

In this paper we will discuss a quartic melonic model with only one interaction. In this case one uses only one intermediate matrix field 
and (dropping the index $c$) the partition function simplifies to:
\begin{align} \label{eq:action}
Z(g) = \int [dH] \;  e^{-N^{D-1} S(H) } \;, \qquad  S(H) =  \frac{1}{2}  \Tr[H^2] -\Tr\left[  \ln \left( \frac{1}{1-g H} \right) \right] \;.
\end{align}

Observe that in this form a tensor models in rank $D$ looks very similar to a matrix model for a $N \times N$ matrix. The only difference comes from the 
modified scaling factor $N^{D-1}$ instead of $N$ for the action. This apparently innocuous modification leads to numerous 
consequences which we will be exploring in detail in this paper.

\section{Vacua of the theory}\label{sec:vacua}

The equations of motion of the model defined in Eq.~\eqref{eq:action} are:
\be \label{eq:eom-1}
H-\frac{g}{1-gH}=0 \;.
\ee
Due to the unitary invariance, Eq.~\eqref{eq:eom-1} fixes only the eigenvalues $\lambda_i$ of $H$: upon diagonalization, the equations of motion become:
\be \label{eq:eom-1eig}
\lambda_i-\frac{g}{1-g\lambda_i}=0 \;\;\; \text{for} \;\;\;i = 1, \ldots, N \, ,
\ee
whose two solutions are:
\be \label{eq:apm}
a_\pm \equiv \frac{1\pm \sqrt{1-4 g^2}}{2g} \, .
\ee
We have a critical value of the coupling at $g_c=1/2$, where we meet a square root singularity and the two solutions merge.
The solution $a_-$ has a regular Taylor expansion around $g=0$ (starting at linear order, $a_- = g +O(g^3)$), while the solution $a_+$ has a simple pole at the origin.

Up to conjugation by an arbitrary unitary matrix, any solution of the equations of motion (i.e. a \emph{vacuum} of the theory) is of the form:
\be \label{vacua}
\bar H_{(N_+,N_-)} = \begin{pmatrix}
                 a_+ \mathbb{I}_+ & 0 \\ 0 & a_-\mathbb{I}_- 
                 \end{pmatrix} = a_+ \begin{pmatrix} \mathbb{I}_+ & 0 \\ 0 & 0 \end{pmatrix} 
                               + a_-  \begin{pmatrix} 0 & 0 \\ 0 &  \mathbb{I}_- \end{pmatrix}
 \equiv a_+ P_+ + a_- P_-  \;,
\ee
where $\mathbb{I}_+$ and $\mathbb{I}_-$ are the identity matrices of size $N_+\times N_+$ and respectively $N_-\times N_-$, with $N_+ +N_-=N$, otherwise arbitrary. There are 
$N+1$ possible solutions of this type.

Observe that from the onset we have neglected the contribution of the measure $[dH]$ to the equations of motion. 
In fact this contribution can profoundly alter the equations of motion and it is possible that the ``vacua'' we identified by just considering
the action are nowhere near the true vacua of the theory. As we will see below, this is in fact what happens in the case of matrices $D=2$.
However, for tensors ($D\ge 3$) the contribution of the measure is subleading in $1/N$ and the solutions in Eq.~\eqref{vacua} are 
true vacuum states in the large  $N$ limit. 

Each $\bar H_{(N_+,N_-)}$ is invariant under the residual symmetry group ${\cal U} (N_+) \times {\cal U} (N_-)$:  
\be
  \bar H_{(N_+,N_-)} = \begin{pmatrix}
                    U_+ & 0 \\ 0 & U_-
                   \end{pmatrix}   \bar H_{(N_+,N_-)} \begin{pmatrix}
                    U_+^{\dagger} & 0 \\ 0 & U_-^{\dagger}
                   \end{pmatrix} \;, \quad \forall U_+ \in {\cal U} (N_+) \;,  \forall U_-\in  {\cal U} (N_-) \; ,
\ee
and, being a representative in a conjugacy class of solutions of the classical equations of motions, is a saddle point of the action $S(H)$. 
The partition function is as a sum over the saddle points $ \bar H_{(N_+,N_-)}$:
\be
  Z(g) = \sum_{N_+} C_{(N_+,N_-)} Z_{(N_+,N_-)} \qquad Z_{(N_+,N_-)} = \int_{H \text{ close to } \bar H_{(N_+,N_-)}} [dH] \; e^{ -N^{D-1} S(H) } \;.
\ee

In the rest of this paper we will discuss the path integral near a saddle point $  \bar H$ (from here on we drop the indices $(N_+,N_-)$ in order to simplify the notation).

Evaluating the action at a generic solution $\bar H$, we obtain:
\be \label{FreeEn}
S(\bar H)=\frac{1}{2}  \Tr\left[\bar H^2\right] -\Tr \left[ \ln \left(  \frac{1}{g}\bar H \right)\right]
 = N_+ \left( \frac{a_+^2}{2} - \ln \frac{a_+}{g} \right) + N_- \left( \frac{a_-^2}{2} - \ln \frac{a_-}{g} \right) \, .
\ee
Near $g=0$, regularity of the free energy demands $N_+=0$. However, at $g=g_c=1/2$, the free energy develops a singularity, and beyond that, solutions with generic $N_+$ are possible. For $g>g_c$, 
it can be easily checked that the real part of \eqref{FreeEn} is independent of $N_+$ and $N_-$. Therefore all these vacua contribute to the same leading exponential order, and in the broken phase we 
do not have at this level any reason to choose a particular value of $N_+$.

Close to a vacuum $\bar H$, we make the change of variables $H= \bar H  + M$, splitting the field into a background $\bar H$ and a fluctuation field $M$. 
Observe that even if $\bar H$ is not hermitian (which happens if $a_{\pm}$ are not real), the perturbation $M$ remains hermitian, and the translation by $\bar H$ is 
a translation of the diagonal entries of $H$ in the complex plane.
Using the equations of motion we obtain:
\begin{align*}
& S(\bar H + M)  = \frac{1}{2}  \Tr[(\bar H + M)^2] -\Tr \left[ \ln \left( \frac{1}{1-g (\bar H + M)} \right) \right]\crcr
& \qquad  = \frac{1}{2}  \Tr[ \bar H^2] + \Tr[\bar H M ] + \frac{1}{2}\Tr[M^2]-\Tr \left[  \ln \left( \frac{1}{1-g \bar H } \right) \right]
- \Tr \left[  \ln \left( \frac{1}{1- \frac{g}{1 -g \bar H} M } \right) \right] \crcr
& \qquad = \frac{1}{2}  \Tr[ \bar H^2]  -\frac{1}{g} \Tr[\bar H ] +  \Tr[\bar H M ] + \frac{1}{2}\Tr[M^2]
- \Tr \left[  \ln \left( \frac{1}{1-  \bar H  M } \right) \right] \crcr
& \qquad = \frac{1}{2}  \Tr[ \bar H^2]  -\frac{1}{g} \Tr[\bar H ] +
 \frac{1}{2} \Bigg[ \Tr[M^2] - \Tr\left[ \bar HM\bar HM  \right]  \Bigg] - \sum_{p\ge 3} \frac{1}{p} 
  \Tr\left[ (\bar HM)^p    \right] \; ,
\end{align*}
therefore the effective quadratic term (defining the inverse covariance) for the fluctuation field is:
\be
 \frac{1}{2}  \Tr\left[M^2\right] - \frac{1}{2} \Tr[ \bar H M \bar H M] .
\ee
Introducing the following definitions:
\be
M^{++} = P_+ M P_+ \, \quad M^{+-} = P_+ M P_- , \quad M^{--} = P_- M P_- , \qquad M = \begin{pmatrix} M^{++}  &  M^{+-}  \\ M^{-+} & M^{--} \end{pmatrix} \;,
\ee
and taking into account that $P_+ + P_- = \mathbb{I}$, the quadratic part for the fluctuation field rewrites as:
\be
 \frac{1}{2} (1- a_+^2 ) \Tr [(M^{++})^2] + \frac{1}{2} (1- a_-^2 ) \Tr [(M^{--})^2] + (1- a_+ a_- ) \Tr [M^{+-} M^{-+}] \;.
\ee
The inverse covariance is an operator in the vectors space of hermitian matrices with inner product $\Braket{M_1|M_2}=\Tr[M_1M_2]$. Since:
\be
  ( P_+ \otimes P_+ ) \ket{M} = \ket{ P_+ M P_+ } = \ket{M^{++}} \;,
\ee
the inverse covariance is:
\be
(1- a_+^2 ) P_+ \otimes P_+  +  (1- a_-^2 ) P_- \otimes P_-  +  (1- a_+ a_- ) (P_+ \otimes P_- + P_- \otimes P_+) \;,
\ee
and $P_+ \otimes P_+$, $P_- \otimes P_-$ and $P_+ \otimes P_- + P_- \otimes P_+$ are (mutually orthogonal) projectors.
The mass eigenvalues for the fluctuation field $M$ are then read off as $\Lambda_+\equiv 1- a_+^2$ with degeneracy $N_+^2$, $\Lambda_-\equiv 1- a_-^2$ with degeneracy $N_-^2$, and $1- a_+ a_- =0$ 
with degeneracy $N_+ N_-$.

Defining $\epsilon= 1-4g^2$ we have $\Lambda_\pm = 2 \sqrt{\epsilon} / (\sqrt{\epsilon}\mp 1)$, and we distinguish the following situations:

\begin{itemize}
 \item For small $\epsilon>0$, the leading behavior of the mass eigenvalues of the fluctuation field is $\Lambda_\pm \sim \mp 2 \sqrt{\epsilon}$, and the only positive eigenvalue is $\Lambda_-$.
 It follows that for $g<g_c$ only the solution $N_-=N, N_+=0$ is stable, and therefore we are in the ${\cal U}(N)$ symmetric phase, in agreement with the small $g$ analysis.

\item At $\epsilon=0$ (i.e. at $g=g_c$), all the mass eigenvalues are zero, signaling criticality.

\item For small $\epsilon<0$, the leading behavior is  $\Lambda_\pm \sim 2 |\epsilon | \mp i 2\sqrt{|\epsilon |}$, and 
both $\Lambda_+$ and $\Lambda_-$ have a positive real part. It follows that for negative $\epsilon$ both vacua are stable and therefore we are in the broken phase, with residual symmetry  ${\cal U} (N_+) \times {\cal U} (N_-)$.
In this case the off-diagonal fluctuations $M_{+-}$ and $M_{-+}$, lacking a quadratic term, have the interpretation of massless Goldstone modes. 

\end{itemize}

Around a vacuum state the effective theory for the fluctuation field writes:
\begin{align} \label{eq:eff-off}
& Z_{(N_+,N_-)} \simeq  e^{- N^{D-1} \left[ N_+ \left( \frac{a_+^2}{2} - \ln \frac{a_+}{g} \right) + N_- \left( \frac{a_-^2}{2} - \ln \frac{a_-}{g} \right)   \right]} \int [dM^{+-} dM^{-+ }] \int [dM^{++}] [dM^{--}] \;\; e^{-N^{D-1} S^{\rm eff} }\crcr
& S^{\rm eff} =  \frac{1}{2}  (1-a_+^2) \Tr[(M^{++})^2] + \frac{1}{2}  (1-a_-^2) \Tr[(M^{--})^2] -\sum_{p\ge 3} \frac{1}{p} \sum_{i_1,\dots i_p\in\{+,-\}} \prod_{q=1}^p a_{i_q} M^{i_qi_{q+1}} \, .
\end{align}

The theory thus obtained looks quite complicated.  We can considerably simplify it by getting rid of the off-diagonal blocks $M^{+-}$ and  $M^{-+}$ (which only appear at cubic order in the above action),
exploiting the unitary invariance. The procedure we detail below comes to explicitly integrating out the Goldstone modes, and obtain the effective theory for the ``radial'' degrees of freedom.

\section{Integrating out the massless modes.}\label{sec:intout}

Knowing that at the saddle point the eigenvalues split in two sets of $N_+$ and $N_-$ identical eigenvalues,
it is convenient to consider the subgroup $\cH = \cU(N_+) \times \cU(N_-)\subset \cU(N)$ defined by:
\be
\cH = \bigg\{ W \in \cU(N) \bigg{| } \;   W=e^{i  \begin{pmatrix}   W_+ & 0 \\ 0 & W_-     \end{pmatrix}  }
 = \begin{pmatrix}   e^{i  W_+} & 0 \\ 0 & e^{i  W_-}      \end{pmatrix}  \bigg\} \;, 
\ee
for two arbitrary Hermitian matrices $W_+$ and $W_-$ of sizes $N_+ \times N_+$ and respectively $N_- \times N_-$.

We first show that any hermitian matrix $H$ can be brought into block diagonal form:
\be
H = V \begin{pmatrix}   H_+ & 0 \\ 0 & H_-  \end{pmatrix} V^\dagger \;,
\ee
for some $V\in \cU(N)/\cH$, i.e. for some $V$ in the left coset of $\cH$ in $\cU(N)$.
We subsequently perform the change of variables $H\to (H_+,H_-,V)$, and compute the Jacobian $[dH] = \det(J)  [dH_+] [dH_-] [dV] $.

\paragraph{Block diagonalization.}
Let us write $H=U E U^\dagger$, where $E$ is diagonal and $U\in \cU(N) $ is a unitary transformation that diagonalizes $H$.
Let us denote $V\in \cU(N)/\cH $ the representative in $ \cU(N)/\cH $ of the equivalence class of $U$, $[U] = \{ U W |\; W \in \cH \}$.
Then there exists a $W \in \cH$ such that\footnote{For the sake of precision, $U\in  \cU'(N)\equiv  \cU(N)/\cU(1)^{\otimes N}$, because a multiplication of $U$ to the right
by a diagonal matrix $\text{ diag}(e^{i\phi_1},\ldots, e^{i\phi_N})$ leaves $H$ unaffected. Similarly, we should also quotient $\cH\to \cH'\equiv \cH/\cU(1)^{\otimes N}$. However, $ \cU(N)/\cH \simeq \cU'(N)/\cH'$.}:
\be
H = UE U^\dagger = VW^{\dagger} E W V^\dagger 
 = V \begin{pmatrix}   e^{- i  W_+} E_+ e^{ i  W_+} & 0 \\ 0 & e^{ - i  W_-} E_- e^{ i  W_-}     \end{pmatrix} V^\dagger  \;,
\ee
where $E_+$ denotes the first $N_+$ eigenvalues of $H$ and $E_-$ denotes the subsequent $N_-$ eigenvalues. 
Defining $H_+ = e^{i  W_+} E_+ e^{-i  W_+}$ and $H_- = e^{i  W_-}E_- e^{-i  W_-}$, shows that $H$ can be brought into block diagonal form 
by conjugation by some $V \in \cU(N)/\cH $. 

\paragraph{Jacobian.} In order to compute the Jacobian $J$ of the change of variables $H\to (H_+,H_-,V)$, we will first 
show that the Jacobian does not depend on $V$ and then compute it for $V$ close to the identity.

Denoting $W_+$ and $W_-$ two Hermitian matrices, and $B$ a complex matrix (of sizes $N_+ \times N_+$, $N_- \times N_-$ and $N_+ \times N_-$, respectively),
any $U\in \cU(N)$ close to the identity writes as:
\be U=e^{i   \begin{pmatrix}  W_+ & B \\ B^\dagger & W_-     \end{pmatrix}  } = 
   e^{i   \begin{pmatrix}   0 & B \\ B^\dagger & 0    \end{pmatrix}  } \; 
    e^{i \begin{pmatrix}   W_+ & 0 \\ 0 & W_-     \end{pmatrix}  } + O(BW) \;,
\qquad   e^{i \begin{pmatrix}   W_+& 0 \\ 0 & W_-     \end{pmatrix}  } \in \cH \;, \ee
hence the class of $U$ in $\cU(N)/\cH$ can be represented by: 
\be
 V   \approx \mathbb{I} +i   \begin{pmatrix}   0 & B \\ B^\dagger & 0     \end{pmatrix} 
   \; .
\ee
The infinitesimal variation of $V$ is then:
\be
 dV =  i  \begin{pmatrix}
        0 & dB \\ dB^{\dagger} & 0 
      \end{pmatrix} \;,
\ee
and for any fixed $W \in \cU(N)$ close to the identity, $ W =  \mathbb{I} +i   T $ (with $T$ a hermitian matrix), we have:
\be
 W V \approx  \mathbb{I} +i T + i   \begin{pmatrix}   0 & B \\  B^{\dagger} & 0     \end{pmatrix} + O(T^2,B^2,QB) \Rightarrow [d(WV)] \approx
 i \begin{pmatrix}   0 & dB \\ dB^{\dagger} & 0     \end{pmatrix} \approx [dV] \;,
\ee
hence the Jacobian does not depend on $V$. In fact, the invariance of the measure induced on the quotient space by the Haar measure on the group can be proved under quite general 
conditions (which in particular apply to our case as $\cU(N)$ is a compact group), see for example \cite{Nachbin}.

We now follow the standard strategy. First, we write an infinitesimal variation of the matrix $H$ as:
\be
 dH = dV \begin{pmatrix}   H_+ & 0 \\ 0 & H_-     \end{pmatrix} V^\dagger
 + V \begin{pmatrix}   H_+ & 0 \\ 0 & H_-     \end{pmatrix} dV^\dagger 
 + V \begin{pmatrix}   dH_+ & 0 \\ 0 & dH_-     \end{pmatrix} V^\dagger \;,
\ee
and use $V V^\dagger = \mathbb{I}  \Rightarrow dV\, V^\dagger + V\,  dV^\dagger = 0$ 
to eliminate $dV^\dagger$. Then we use the invariance of the measure and evaluate the Jacobian at $V=\mathbb{I}$:
\be
 dH =   i  \begin{pmatrix}    0 & dB H_- \\ dB^\dagger H_+ & 0      \end{pmatrix}
 -  i   \begin{pmatrix}    0 & H_+ dB \\  H_- dB^\dagger & 0      \end{pmatrix}
 + \begin{pmatrix}   dH_+ & 0 \\ 0 & dH_-     \end{pmatrix} .
\ee
Finally, writing $dH_{ij}= J_{ij,kl} dX_{kl}$, for $X=(H_+,H_-,B,B^\dagger)$ and $H = (H^{++}, H^{--}, H^{+-}, H^{-+})$ we read off the Jacobian:
\be
 J \equiv \frac{dH}{dX} =  \begin{pmatrix}    \mathbb{I}_+\otimes \mathbb{I}_+ & 0 & 0 & 0 \\
                            0 & \mathbb{I}_-\otimes \mathbb{I}_- & 0 & 0 \\
                            0 & 0 &  i ( \mathbb{I}_+\otimes H_- - H_+\otimes \mathbb{I}_- ) & 0 \\
                            0 & 0 & 0 &  -i (  H_- \otimes \mathbb{I}_+ -  \mathbb{I}_-  \otimes H_+)
  \end{pmatrix} \;,
\ee
with determinant $\det(J)= \det( H_+\otimes \mathbb{I}_- -\mathbb{I}_+\otimes H_-)^2$.
We have thus shown that:
\begin{equation} \label{eq:block-dec}
\boxed{  \int [dH] \;  e^{-N^{D-1}S(H)} = \int [dV] \int [dH_+] [dH_-] \;\;   \Big[  \det \left(  H_+ \otimes \mathbb{I}_- - \mathbb{I}_+ \otimes H_-\right) \Big]^2 \; e^{ -N^{D-1}S \begin{pmatrix}    H_+ & 0 \\ 0 & H_-      \end{pmatrix}   } \; .}
\end{equation}
The integral over $V$ is now completely decoupled and can be absorbed into a proper normalization.

This formula generalizes the usual reduction to eigenvalues, which can be recovered by taking $N_-=1$, and then iteratively repeating the decomposition on the remaining $N_+\times N_+$ block.
This will iteratively generate the usual Vandermonde determinant $\Delta(\lambda)^2= \prod_{i<j} (\lambda_i-\lambda_j)^2$.
Keeping instead $N_+$ and $N_-$ generic, but reducing to eigenvalues our formula \eqref{eq:block-dec}, we write:
\begin{equation} \label{eq:block-eig}
 \int [dH] \;  e^{-N^{D-1}S(H)} = \int \left( \prod_{i=1}^{N_+} d\lambda^+_i \, \Delta(\lambda^+)^2\, e^{ -N^{D-1}S(\lambda^+_i)}\right) 
     \left( \prod_{j=1}^{N_-} d\lambda^-_j \, \Delta(\lambda^-)^2 \, e^{ -N^{D-1}S(\lambda^-_j)} \right) \;\; 
        \prod_{i,j}  (\lambda^+_i -\lambda^-_j)^2 \;,
\end{equation}
and of course the two sets of eigenvalues can be recombined as $\lambda=(\lambda^+,\lambda^-)$ to give the usual expression:
\begin{equation} \label{eq:usual-eig}
 \int [dH] \;  e^{-N^{D-1}S(H)} = \int \prod_{i=1}^{N} d\lambda_i \, \Delta(\lambda)^2\, e^{ -N^{D-1}S(\lambda_i)} \, .
\end{equation}

The effect of the scaling with $N^{D-1}$ in Eq.~\eqref{eq:action} comes now into play. 
From Eq.~\eqref{eq:usual-eig} we conclude that we can evaluate the partition function by a saddle point method without having to care about the Vandermonde determinant if $D>2$. In fact $\log \Delta(\lambda)$ is of order $N^2$, 
whereas $N^{D-1}\sum_i S(\lambda_i)$ is of order $N^D$.
This is de facto what we have done when writing the saddle point equation \eqref{eq:eom-1eig}  for the action alone (i.e. without measure contributions). Eq. \eqref{eq:eom-1eig} yields just two solutions 
for the eigenvalues at the saddle point. We can then directly use eq. \eqref{eq:block-eig} and perform the saddle point approximation, taking in each block one of the solutions
i.e. $\lambda^+=a_+$ and $\lambda^-=a_-$.
In order to study the fluctuations around the saddle point one would then write $\lambda^\pm_i = a_\pm + \mu^\pm_i$ and study the corrections coming from the diagonal fluctuations $\mu^\pm_i$.

\section{The effective theory}\label{sec:efft}

If one computes the partition function using the eigenvalue decomposition, the connection between the Feynman graphs of the theory and random triangulations is lost. In order to keep this connection explicit
and find an interpretation of the effective theory around a non trivial vacuum in terms of random triangulations, one must 
start from Eq.~\eqref{eq:block-dec}, and decompose the matrices $H_{\pm}$ into a background and a fluctuation $H_\pm = a_\pm + M_\pm$.
We thus obtain:
\begin{align} \label{eq:eff}
Z_{(N_+,N_-)}  = & e^{- N^{D-1} \left[ N_+ \left( \frac{a_+^2}{2} - \ln \frac{a_+}{g} \right) + N_- \left( \frac{a_-^2}{2} - \ln \frac{a_-}{g} \right)   \right]} \crcr
   & \times \int [dM_+] [dM_-]   \;\;   \Big[ \det \Big( (a_+-a_-)\mathbb{I}_+\otimes \mathbb{I}_-  + M_+ \otimes \mathbb{I}_- - \mathbb{I}_+ \otimes M_-\Big) \Big]^2 \;  e^{-N^{D-1} S^{\rm eff} } \;, \crcr
 S^{\rm eff} =&  \frac{1}{2}  (1-a_+^2) \Tr[M_+^2] + \frac{1}{2}  (1-a_-^2) \Tr[M_-^2]
 - \sum_{p\ge 3} \frac{a_+^p}{p}   \Tr\left[ M_+^p    \right] - \sum_{p\ge 3} \frac{a_-^p}{p}   \Tr\left[ M_-^p    \right]  \, ,
\end{align}
achieving our goal of eliminating the off-diagonal blocks in Eq.~\eqref{eq:eff-off}, at the price of introducing a non trivial determinant term.
Note however that $M_\pm$ in Eq.~\eqref{eq:eff} are not the same matrices as the $M^{++}$ and $M^{--}$ in Eq.~\eqref{eq:eff-off}, in particular they depend non-trivially on the off-diagonal blocks $M^{+-}$ and $M^{-+}$ in 
Eq.~\eqref{eq:eff-off} just like the eigenvalues of a matrix depend on its off-diagonal elements. 

\subsection{Feynman graphs}

The integral in Eq.~\eqref{eq:eff} for the fluctuations around a saddle point:
\begin{align*}
  &   (a_+-a_-)^{2N_+N_-}  \int [dM_+] [dM_-]   \;\;  
   e^{2 \Tr \left[ \ln \left( \mathbb{I}_+\otimes \mathbb{I}_- + \frac{1}{a_+-a+_-} ( M_+ \otimes \mathbb{I}_- - \mathbb{I}_+ \otimes M_-  ) \right) \right] } \crcr
& \;\; \times  e^{-N^{D-1} \left[   \frac{1}{2}  (1-a_+^2) \Tr[M_+^2] + \frac{1}{2}  (1-a_-^2) \Tr[M_-^2]
 - \sum_{p\ge 3} \frac{a_+^p}{p}   \Tr\left[ M_+^p    \right] - \sum_{p\ge 3} \frac{a_-^p}{p}   \Tr\left[ M_-^p    \right]    \right] } \; ,
\end{align*}
can be written more explicitly by expanding the logarithm:
\begin{align*}
 & (a_+-a_-)^{2N_+N_-} \int [dM_+] [dM_-]   \;\;  
   e^{2 \sum_{q\ge 1} \frac{(-1)^{q+1}}{q (a_+-a_-)^q }  \sum_{q_-=0}^q \binom{q}{q_-} (-1)^{q_-}\Tr[M_+^{q-q_- }] \Tr[M_-^{q_-}]  } \crcr
& \;\; \times  e^{-N^{D-1} \left[   \frac{1}{2}  (1-a_+^2) \Tr[M_+^2] + \frac{1}{2}  (1-a_-^2) \Tr[M_-^2]
 - \sum_{p\ge 3} \frac{a_+^p}{p}   \Tr\left[ M_+^p    \right] - \sum_{p\ge 3} \frac{a_-^p}{p}   \Tr\left[ M_-^p    \right]    \right] } \; ,
\end{align*}
leading to the effective action:
\begin{align} \label{Seff2}
& \tilde S^{\rm eff} = N^{D-1}  \frac{1}{2}  (1-a_+^2) \Tr[M_+^2]+ N^{D-1} \frac{1}{2}  (1-a_-^2) \Tr[M_-^2] \crcr
&\qquad  -  N^{D-1} \sum_{p_+\ge 3} \frac{a_+^{p_+}}{p_+}   \Tr\left[ M_+^{p_+}    \right] - N^{D-1}\sum_{p_-\ge 3} \frac{a_-^{p_-} }{p_-}   \Tr\left[ M_-^{p_-}    \right]    \crcr
& \qquad +  \sum_{q_+,q_-\ge 0}^{q_++q_-\ge 1} \frac{2}{ (a_+ -a_-)^{q_++q_-} } \frac{(-1)^{q_+}}{q_++q_-} \binom{q_++q_-}{q_+}  \Tr\left[ M_+^{q_+}    \right]  \Tr\left[ M_-^{q_-}    \right]  \; .
\end{align}

Observe that the terms coming from the non trivial determinant generate new quadratic terms. These terms could be included in the measure, but we prefer to treat them as bivalent vertices.

The logarithm of the partition function is a sum over connected graphs. The graphs have three kinds of vertices:
\begin{itemize}
 \item vertices coming from $\sum_{p_+\ge 3} \frac{a_+^{p_+}}{p_+}   \Tr\left[ M_+^{p_+}    \right]  $. We denote $v_{p_+}$ the number of such vertices of coordination $p_+$ and $V_+$ the total number of such vertices,
    $V_+ = \sum_{p_+\ge 3} v_{p_+}$. These vertices have coordination at least 3. These vertices represent polygons in the $+$ sector, which will glue together to form $+$ surfaces. 
 \item vertices coming from $\sum_{p_-\ge 3} \frac{a_-^{p_-} }{p_-}   \Tr\left[ M_-^{p_-}    \right]  $. We denote $v_{p_-}$ the number of such vertices of coordination $p_-$ and $V_-$ the total number of such vertices,
    $V_- = \sum_{p_-\ge 3} v_{p_-}$. These vertices have coordination at least 3. These vertices represent polygons in the $-$ sector, which will glue together to form $-$ surfaces. 
 \item vertices coming from $ \sum_{q_+,q_-\ge 0}^{q_++q_-\ge 1} \frac{2}{ (a_+ -a_-)^{q_++q_-} } \frac{(-1)^{q_+}}{q_++q_-} \binom{q_++q_-}{q_+}  \Tr\left[ M_+^{q_+}    \right]  \Tr\left[ M_-^{q_-}    \right]  $.
 We denote $v_{q_+q_-}$ the number of such vertices of coordination $q_+$ and $q_-$ and $V_{+-}$ the total number of such vertices $V_{+-}=\sum_{q_+,q_-\ge 0}^{q_++q_-\ge 1} v_{q_+q_-}$. These vertices have a joint
 coordination $q_++q_-$ at least one, but can have coordination $0$ in one of the sectors. These vertices represent a pair of a $+$ and a $-$ polygons connected at a nodal point. The $+$ polygons (resp. the $-$ polygons)
 will glue into $+$ (respectively $-$) surfaces, and the $+$ and $-$ surfaces will be connected at a mixed vertex by a nodal point.
\end{itemize}

We can represent the $+$ and $-$ vertices as solid vertices, and the $+-$ vertices as two solid vertices (one for the $+$ sector and one for the $-$ sector) connected by a dashed edge, see Figure \ref{fig:graph}. 
Each graph represents a nodal surface (for definitions and properties of nodal surfaces see for example \cite{petro2008moduli}). 
\begin{figure}[ht]
\begin{center}
 \includegraphics[width=4cm]{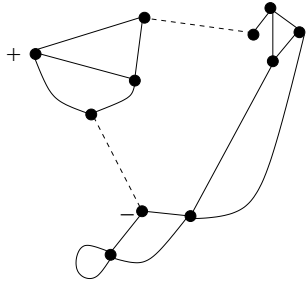}  
\caption{An example of a graph generated by the effective action in Eq.~\eqref{Seff2}. Here the two dashed edges represent two $+-$ vertices 
of the type $\Tr \left[ M_+^2 \right]  \Tr\left[ M_-^2  \right]$ and  $\Tr \left[ M_+^2 \right]  \Tr\left[ M_-  \right]$.}
\label{fig:graph}
\end{center}
\end{figure}
The $E_+$ edges connecting two $M_+$ half edges (which we call the $+$ edges) bring a factor $\frac{1}{N^{D-1} (1-a_+^2) }$. 
The $E_-$ edges connecting two $M_-$ half edges (which we call the $-$ edges) bring a factor $ \frac{1}{N^{D-1} (1-a_-^2) }$. 

The graph is connected but, by erasing the dashed edges, it might disconnect into several 
connected components. Let us call $C_+$ the number of the connected components of the graph made exclusively of the edges $+$, $g_+$ the sum of the genera of these components and $F_+$ the number of faces of these 
components (and similarly for the $-$ components). Up to irrelevant combinatorial factors the amplitude of a graph is:
\begin{align}
 & \left( \frac{1}{ N^{D-1} (1-a_+^2) } \right)^{E_+}   \left( \frac{1}{ N^{D-1} (1-a_-^2) } \right)^{E_-}
 N^{(D-1) V_+ } a_+^{\sum_{p_+\ge 3} p_+v_{p_+} }  N^{(D-1)V_- } a_-^{\sum_{p_-\ge 3} p_- v_{p-} }\crcr
 & \times  \frac{(-1)^{ \sum_{q_+,q_-\ge 0}^{q_++q_-\ge 1} (q_++1) v_{q_+q_-} }}{(a_+-a_-)^{ \sum_{q_+,q_-\ge 0}^{q_++q_-\ge 1} (q_++q_-) v_{q_+q_-} }}  N_+^{F_+} N_-^{F_-} \; .
\end{align}
Note that in the symmetric phase (i.e. $g<g_c$ and $N_+=0$) the amplitudes are positively definite, while in the broken phase (i.e. $g>g_c$ and $N_+>0$) they are complex. 
In the following we will not keep track of global minus signs for individual amplitudes.

The various numbers of vertices, edges and faces are related by the following topological relations: 
\begin{align} \label{DehnSommerville}
& 2 E_+  =  \sum_{p_+\ge 3} p_+ v_{p_+} + \sum_{q_+,q_-\ge 0}^{q_++q_-\ge 1} q_+ v_{q_+q_-} \;, \qquad 2 E_-  =  \sum_{p_-\ge 3} p_- v_{p_-} + \sum_{q_+,q_-\ge 0}^{q_++q_-\ge 1}q_- v_{q_+q_-} \; ,
\end{align}
\begin{align} \label{Euler}
& V_+  + V_{+-} - E_+ + F_+ = 2C_+ -2g_+ \;, \qquad V_- + V_{+-} - E_- + F_- = 2C_- -2g_- \; ,
\end{align}
\begin{align} \label{Betti}
& E_+ - V_+ - V_{+-} + C_+ \ge 0 \;, \qquad E_- - V_- -V_{+-} + C_- \ge 0 \; , \crcr
&  E_+ + E_-  - V_+ -V_- -V_{+-} +1 \ge 0 \; ,\crcr
& V_{+-} - C_+ - C_- +1\ge 0 \; .
\end{align}
Equations \eqref{DehnSommerville} and  \eqref{Euler} are usual graph relations applied to the $+$ and $-$ subgraphs:
the first simply states that summing the numbers of half edges at every vertex amounts to counting twice the number of edges, while the second is simply the Euler characteristic formula.
Equations \eqref{Betti} state that the first Betti number of a graph is always non-negative, in particular it is zero for trees and positive for graphs with loops (in fact it counts the number of loop edges, or \emph{excess edges}).
A generic connected graph, consisting of $+$ and $-$ components and dashed edges, is dual to a nodal surface with $V_{+-}$ nodal points. 
The total genus of a nodal surface is $g_+ +g_- + (  V_{+-} -C_+ - C_- +1 )$, i.e. the sum of the genera of the $+$ components, the $-$ components and the first Betti 
number of the abstract graph obtained by collapsing the $+$ and $-$ components to vertices connected by the dashed edges. Of course $g_+ +g_- + (  V_{+-} -C_+ - C_- +1 )\geq 0$, and 
the equality holds only for the case in which all the $+$ and $-$ component are planar, and the abstract graph is a tree.

Defining the filling fractions $n_+ = \frac{N_+}{N} \le 1 $ and $n_- = \frac{N_-}{N} \le 1$, the amplitude of a graph writes:
\begin{align*}
 & \frac{a_+^{\sum_{p_+\ge 3} p_+v_{p_+} } a_-^{\sum_{p_-\ge 3} p_- v_{p-} }   }{ (1-a_+^2)^{E_+}   (1-a_-^2)^{E_-}
 (a_+-a_-)^{ \sum_{q_+,q_-\ge 0}^{q_++q_-\ge 1} (q_++q_-) v_{q_+q_-} }  } \; n_+^{F_+} n_-^{F_-} \crcr
 & \qquad \times  N^{  (D-1) V_+ - (D-1) E_+ +F_+   }  N^{  (D-1) V_+ - (D-1) E_+ +F_-   }  \; ,
\end{align*}
which, using:
\begin{align*}
& (D-1) V_+ - (D-1) E_+ +F_+  = (D-1) V_+ - (D-1) E_+ + E_+ - V_+ -V_{+-} + 2 C_+ -2g_+  \crcr
& \qquad =  - (D-2) (E_+ -V_+   )  -V_{+-} + 2 C_+ -2g_+ \;,
\end{align*}
becomes:
\begin{align} \label{graphAmp}
  & \frac{a_+^{\sum_{p_+\ge 3} p_+v_{p_+} } a_-^{\sum_{p_-\ge 3} p_- v_{p-} }   }{ (1-a_+^2)^{E_+}   (1-a_-^2)^{E_-}
 (a_+-a_-)^{ \sum_{q_+,q_-\ge 0}^{q_++q_-\ge 1} (q_++q_-) v_{q_+q_-} } } \; n_+^{F_+} n_-^{F_-} \crcr
 & \qquad \qquad  \times N^{  
  -(D-2) (  E_+ - V_+  ) - (D-2) ( E_- - V_-   ) + 2 - 2 ( V_{+-} - C_+ - C_- +1)  -2g_+ -2g_-
 }\; .
\end{align}

\section{The phases of the model and their geometry}\label{sec:phgeom}

The main difference between the $D=2$ case of matrices and the $D\ge 3$ case of tensors comes from the following theorem.
\begin{theorem}
 The number of graphs at fixed order in $1/N$ is:
 \begin{itemize} 
  \item \emph{finite} for $D\ge 3$,
  \item \emph{infinite} for $D=2$. 
 \end{itemize}
\end{theorem}
\begin{proof}
Let us consider the scaling with $N$ of a term:
\begin{align*}
 -(D-2) (  E_+ - V_+  ) - (D-2) ( E_- - V_-   ) + 2 - 2 ( V_{+-} - C_+ - C_- +1)  -2g_+ -2g_- \;.
\end{align*}
As the graph is connected $ V_{+-} - C_+ - C_- +1 \ge 0 $, hence this scaling is bounded from above by:
\begin{align*}
   2 - (D-2) (  E_+ - V_+  ) - (D-2) ( E_- - V_-   ) \; .
\end{align*}
The $V_+$ and $V_-$ vertices are at least three valent, hence $ 2E_+ \ge 3V_+$ and $2E_- \ge 3V_-$, and the scaling with $N$ of a term is bounded from above by:
\[
 2 - \frac{D-2}{3} (E_+ + E_- ) \; .
\]
Finally, the number of dashed edges equals $V_{+-}$ and is itself bounded by $E_+ + E_- + 1$. Consequently, for $D\ge 3$, at each fixed scaling in $1/N$ only 
connected graphs with at most a finite number of ($+$, $-$, and dashed) edges contribute. As the number of such graphs is finite the first assertion of the theorem follows.

For $D=2$ the above bound does not constrain the number of graphs (it just states that the amplitude of any graph is bounded by $N^2$). The scaling with $N$ of a graph becomes for $D=2$:
\[
 2 - 2 \left( V_{+-} -C_+ - C_- +1 + g_+ +g_-    \right) \;,
\]
i.e. it is proportional with the total genus of the nodal surface dual to the graph. As the number of triangulated surfaces of fixed total genus is infinite, the second assertion of the theorem follows.

\end{proof}

Having estimated a bound for the scaling in $N$ of the graphs generated by the effective action Eq.~\eqref{Seff2}, we can now address the question whether our expansion around the ``vacua'' in Eq.~\eqref{vacua} is justified.

The connected one point function $ \Braket{ \frac{1}{N} \Tr[M_{-}]}_{\rm connected}$, is a sum over graphs having an extra $-$ univalent vertex representing the external field $M_{--}$. 
The scaling with $N$ of such a graph is:
\be
 -1 -(D-1)  -(D-2) (  E_+ - V_+  ) - (D-2) ( E_- - V_-   ) + 2 - 2 ( V_{+-} - C_+ - C_- +1)  -2g_+ -2g_- \;,
\ee
where $V_-$ includes the extra univalent $-$  vertex. As $2E_- \ge 1 + 3(V_--1)$. It follows that the scaling with $N$ of a connected one point graph is bounded by 
$  -(D-2) (  E_+ - V_+  ) - (D-2) \big[ E_- -  ( V_- -1)   \big] \le -\frac{D-2}{3} \left( E_+ + E_- + 1\right)  $ hence:
\be
  \Braket{ \frac{1}{N} \Tr[M_{-}]}_{\rm connected} \sim \begin{cases}
     O(1) \;, & \qquad D=2 \crcr
     O(N^{-\frac{D-2}{3}}) \;, & \qquad D\ge 3
 \end{cases} \;,
\ee
that is the connected one point function starts at order $1$ for $D=2$, but is strictly smaller in $1/N$ for $D\ge 3$.

Similarly, we can look at the free energy:
\be
\cF = -\frac{1}{N^D} \ln Z_{(N_+,N_-)} 
 = n_+ \left( \frac{a_+^2}{2} - \ln \frac{a_+}{g} \right) + n_- \left( \frac{a_-^2}{2} - \ln \frac{a_-}{g} \right) 
   + \frac{1}{N^{D-2}}\tilde \cF \, ,
\ee
where $\tilde \cF$ contains the graphs for which we proved that an upper bound for the scaling is $N^{- \frac{D-2}{3} (E_+ + E_- )}$.
That is, the correction from the effective theory to the free energy computed at the saddle point is of the same order as the latter for $D=2$, but it is subleading in $1/N$ for $D>2$.

\subsection{The matrix case $D=2$} 

It is now apparent why the $D=2$ case is very different from the $D\ge 3$. Indeed, for $D=2$ the critical behavior of the amplitude \eqref{graphAmp} of an individual graph has no 
relevance: the free energy at fixed order in $1/N$ is a sum of an \emph{infinite} family of graphs, and acquires a critical behavior due to the criticality of this infinite sum. In particular, 
it is well known that this sum becomes critical at the critical constant $g'_c = \frac{1}{\sqrt{12} }$, 
which is smaller than the critical value $g_c = \frac{1}{2}$, at which the amplitude of an individual graph becomes critical. 

The fact that one still obtains an infinity of graphs at any fixed order in $1/N$ shows that the contribution coming from the measure (the Vandermonde determinant) can not be neglected in this case.
The ``vacuum state'' of a diagonal matrix with entries $a_+$ and $a_-$ which we obtained by neglecting the contributions of the measure is far from a genuine vacuum state of the theory,
because, as we showed above, $ \Braket{ \frac{1}{N} \Tr[M_{-}]}_{\rm connected}$ is of the same order in $1/N$ as $ \Braket{ \frac{1}{N} \Tr[\bar H]}_{\rm connected}$.

\subsection{The tensor case $D\ge 3$.}

For $D\ge 3$ one can derive the critical behavior of the free energy just by looking at the critical behavior of individual graphs, as the free energy
is just a finite sum over such contributions. The diagonal vacua we identified are close to the genuine vacua of the theory as the connected one point function is subleading in $1/N$.

\subsubsection{The symmetric phase}

We first consider the case $\epsilon = 1-4g^2 > 0$. The only positive mass eigenvalue is $\Lambda_-$, hence the unique stable vacuum 
corresponds to filling fractions $n_+=0, n_-=1$. In this case we have no $+$ components, and no nodal points.
In this phase the graphs are made of a single connected $-$ component, and their amplitude is:
\begin{align*}
  & \frac{   a_-^{\sum_{p_-\ge 3} p_- v_{p-} }   }{    (1-a_-^2)^{E_-}
 } \quad   N^{  2 - (D-2) ( E_- - V_-  )      -2g_- 
 }   \;.
\end{align*}

From the onset we observe that the non planar graphs are subleading. Indeed, if a graph is non planar $g_->0$, its amplitude is strictly 
smaller than the amplitude of a planar graph having the same number of edges $E_-$ and the same number of loop edges $E_- - V_- + 1  $ which is planar.

In the critical regime the amplitude of a planar graphs is:
\be
 \frac{1}{\sqrt{ \epsilon}^{ E_- } }   N^{2 - (D-2) ( E_- - V_-   )  } \;,
\ee
and, as all the $V_-$ vertices are at least three valent, we have $ V_- \le \frac{2}{3} E_- $, therefore this amplitude scales at most like:
\be
  N^2 \left( \sqrt{ \epsilon}  N^{\frac{D-2}{3}} \right)^{-E_-} \; .
\ee
The suppression in $N$ of this amplitude can be offset by the enhancement at criticality.  
In the double scaling regime $ N\to \infty $, $\alpha^{-1} \equiv \sqrt{ \epsilon}  N^{\frac{D-2}{3} } $ fixed, \emph{arbitrary planar three-valent graphs} contribute at the same order and 
dominate the free energy.

This is to be contrasted with the double scaling limit in the symmetric phase for the tensor model which includes all the interactions, Eq.~\eqref{eq:4model}, discussed in \cite{DGR}. In that case,
in the double scaling regime either only trees decorated with self loop edges (tadpole edges) on the leafs contribute for $D < 6$, or arbitrary three valent graphs contribute for $D\ge 6$.
The family of arbitrary three valent graphs is non summable, hence for $D\ge 6$ the double scaling series is not summable, like it is for matrices. For $D<6$, the family 
of trees decorated with self loop edges on the leafs is summable. However, it is very likely that it corresponds to a branched polymer phase.

For tensor models with only one interaction the dominant family in double scaling is summable, but it consists in all the planar three valent graphs. It is most likely that this 
double scaling regime corresponds to a Brownian sphere geometric phase of the model.

The free energy is constructed by summing all the planar three-valent graphs with (up to combinatorial factors) weight $\alpha^{E_-}$, 
therefore we expect to find a finite critical value $\alpha_c$ at which 
graphs with an infinite number of edges (or vertices, i.e. triangles in the dual surface) dominate. This is the standard situation that one exploits to construct a continuum limit, but we will not push this further here.

\subsubsection{The broken phase}

When $\epsilon = 1-4g^2 < 0$ both mass eigenvalues $\Lambda_+, \Lambda_-$ have a positive real part, hence we can have $n_+>0$. 
In this regime the amplitude of a graph is proportional to:
\be
 \frac{n_+^{F_+} n_-^{F_-}}{\sqrt{\epsilon}^{E_+ + E_- +   \sum_{q_+,q_-\ge 0}^{q_++q_-\ge 1} (q_++q_-) v_{q_+q_-}      }}   
   \; N^{  
  -(D-2) (  E_+ - V_+  ) - (D-2) ( E_- - V_-   ) + 2 - 2 ( V_{+-} - C_+ - C_- +1)  -2g_+ -2g_- 
 } \; .
\ee
As in the symmetric phase one can from the onset consider only the graphs such that $g_+$ and $g_-$ are zero. 

Using the topological relations, the amplitude of a graph with planar $+$ and $-$ components is:
\begin{align}
&  \frac{n_+^{F_+} n_-^{F_-}}{\sqrt{\epsilon}^{3E_+ + 3 E_-}} \sqrt{\epsilon}^{  \sum_{p_+\ge 3} p_+ v_{p_+}   \sum_{p_-\ge 3} p_+ v_{p_-}   } 
 N^{    -(D-2) (  E_+ - V_+  ) - (D-2) ( E_- - V_-   ) + 2 - 2 ( V_{+-} - C_+ - C_- +1)  } 
  \; .
\end{align}

Let us consider a graph $\cG$ and suppose that it possesses an internal $V_+$ vertex of coordination $p_+ >3$. We compare the amplitude of this graph with the one of the graph 
$\cG'$ which is identical with $\cG$, except that the $p_+$ valent vertex is replaced by a $3$ valent vertex and a $p_+-1 $ valent vertex connected by a $+$ edge. As $\cG'$ has 
one more $+$ edge and one more $+$ vertex with respect to $\cG$, the scaling in $N$ of the two amplitudes is the same. We have:
\be
 A(\cG') = \frac{1}{\sqrt{\epsilon}^3} \sqrt{\epsilon}^{3 + p_+-1 -p_+} A(\cG) \Rightarrow A(\cG') > A(\cG) \;.
\ee
It follows that the most singular graphs must have only three valent $+$ and $-$ vertices and their amplitude is:
\begin{align}
&  \frac{n_+^{F_+} n_-^{F_-}}{\sqrt{\epsilon}^{3E_+ + 3 E_-}} \sqrt{\epsilon}^{ 3V_+ + 3V_-  } 
 N^{    -(D-2) (  E_+ - V_+  ) - (D-2) ( E_- - V_-   ) + 2 - 2 ( V_{+-} - C_+ - C_- +1)  } \crcr
& \qquad = n_+^{2+E_+/3} n_-^{2+E_-/3} N^2 \left( \sqrt{\epsilon} N^{\frac{D-2}{3}}\right)^{- (E_+  + E_- ) }  N^{ - 2 ( V_{+-} - C_+ - C_- +1)  } \; .
\end{align}

In the double scaling regime $ N\to \infty $, $\alpha^{-1} \equiv \sqrt{ \epsilon}  N^{\frac{D-2}{3} } $ fixed, graphs with:
\begin{itemize}
 \item only three valent $V_+$ and $V_-$ vertices,
 \item planar $+$ and $-$ components,
 \item dashed edges that form a tree connecting the planar $+$ and $-$ components,
\end{itemize}
contribute to the same order and dominate the free energy. Such graphs represents planar nodal surfaces, such that all the vertices which are not nodal points are three valent.

Leading graphs in the double scaling regime in the symmetric and broken phases are represented in Figure \ref{fig:graph1}.
\begin{figure}[ht]
\begin{center}
 \includegraphics[width=8cm]{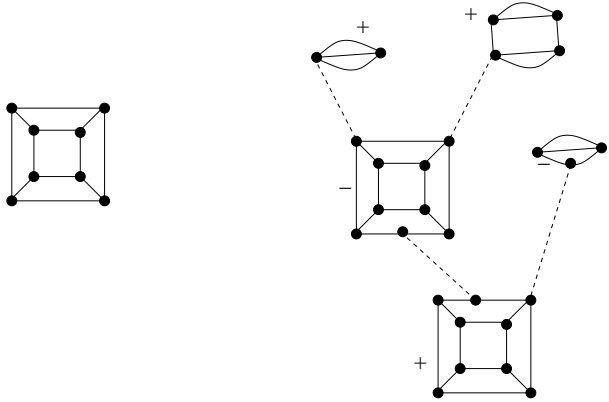}  
\caption{Examples of graphs dominating the double scaling regime. On the left a graph dominating the symmetric phase and on the right a graph dominating the broken phase.}
\label{fig:graph1}
\end{center}
\end{figure}

In the double scaling regime, the free energy in the broken phase is constructed by summing all the graphs detailed above, with a weight proportional to $(n_+^{1/3} \alpha)^{E_+}(n_-^{1/3} \alpha)^{E_-}$. 
Therefore it is effectively a double series in the parameters $z_+=n_+^{1/3} \alpha$ and $z_-=n_+^{1/3} \alpha$, which we expect to have a non-trivial domain of convergence in $\mathbb{C}^2$.
For $n_+=n_-=1/2$ and hence $z_+=z_-=z$, this series becomes again a simple series in one variable, $z$. In this case we expect to find a finite radius of convergence. Note however that the coefficients of the power 
series are now complex, therefore one might have to tune $z$ to a complex value in order to attain a continuum limit with a standard (in the sense of DT) geometrical interpretation.

\section*{Acknowledgements}
We would like to thank St\'ephane Dartois for useful remarks on nodal surfaces.

\providecommand{\href}[2]{#2}\begingroup\raggedright\endgroup


\begin{thebibliography}{10}

\bibitem{ambj3dqg}
J.~Ambjorn, B.~Durhuus, and T.~Jonsson, ``{Three-dimensional simplicial quantum
  gravity and generalized matrix models},''
\href{http://dx.doi.org/10.1142/S0217732391001184}{{\em Mod. Phys. Lett.}
  {\bfseries A6} (1991) 1133--1146}.

\bibitem{sasa1}
N.~Sasakura, ``{Tensor model for gravity and orientability of manifold},''
\href{http://dx.doi.org/10.1142/S0217732391003055}{{\em Mod. Phys. Lett.}
  {\bfseries A6} (1991) 2613--2624}.

\bibitem{color}
R.~Gurau, ``{Colored group field theory},''
  \href{http://dx.doi.org/10.1007/s00220-011-1226-9}{{\em Commun. Math. Phys.}
  {\bfseries 304} (2011) 69--93},
\href{http://arxiv.org/abs/0907.2582}{{\ttfamily arXiv:0907.2582 [hep-th]}}.

\bibitem{review}
R.~Gurau and J.~P. Ryan, ``{Colored tensor models - a review},''
  \href{http://dx.doi.org/10.3842/SIGMA.2012.020}{{\em SIGMA} {\bfseries 8}
  (2012) 020},
\href{http://arxiv.org/abs/1109.4812}{{\ttfamily arXiv:1109.4812 [hep-th]}}.

\bibitem{DiFrancesco:1993nw}
P.~Di~Francesco, P.~H. Ginsparg, and J.~Zinn-Justin, ``{2-D Gravity and random
  matrices},'' \href{http://dx.doi.org/10.1016/0370-1573(94)00084-G}{{\em Phys.
  Rept.} {\bfseries 254} (1995) 1--133},
\href{http://arxiv.org/abs/hep-th/9306153}{{\ttfamily arXiv:hep-th/9306153
  [hep-th]}}.

\bibitem{EDTDavid}
F.~David, ``{Simplicial quantum gravity and random lattices},''
\href{http://arxiv.org/abs/hep-th/9303127}{{\ttfamily arXiv:hep-th/9303127
  [hep-th]}}.

\bibitem{EDTAmbjorn}
J.~Ambjorn, B.~Durhuus, and T.~Jonsson, {\em {Quantum geometry. A statistical
  field theory approach}}.
\newblock Cambridge University Press, Cambridge, 1997.

\bibitem{Kazakov:1985ds}
V.~A. Kazakov, ``{Bilocal regularization of models of random surfaces},''
\href{http://dx.doi.org/10.1016/0370-2693(85)91011-1}{{\em Phys. Lett.}
  {\bfseries B150} (1985) 282--284}.

\bibitem{David:1984tx}
F.~David, ``{Planar diagrams, two-dimensional lattice gravity and surface
  models},''
\href{http://dx.doi.org/10.1016/0550-3213(85)90335-9}{{\em Nucl. Phys.}
  {\bfseries B257} (1985) 45}.

\bibitem{critical}
V.~Bonzom, R.~Gurau, A.~Riello, and V.~Rivasseau, ``{Critical behavior of
  colored tensor models in the large N limit},''
  \href{http://dx.doi.org/10.1016/j.nuclphysb.2011.07.022}{{\em Nucl. Phys.}
  {\bfseries B853} (2011) 174--195},
\href{http://arxiv.org/abs/1105.3122}{{\ttfamily arXiv:1105.3122 [hep-th]}}.

\bibitem{uncoloring}
V.~Bonzom, R.~Gurau, and V.~Rivasseau, ``{Random tensor models in the large N
  limit: Uncoloring the colored tensor models},''
  \href{http://dx.doi.org/10.1103/PhysRevD.85.084037}{{\em Phys. Rev.}
  {\bfseries D85} (2012) 084037},
\href{http://arxiv.org/abs/1202.3637}{{\ttfamily arXiv:1202.3637 [hep-th]}}.

\bibitem{'tHooft:1973jz}
G.~'t~Hooft, ``{A planar diagram theory for strong interactions},''
\href{http://dx.doi.org/10.1016/0550-3213(74)90154-0}{{\em Nucl. Phys.}
  {\bfseries B72} (1974) 461}.

\bibitem{expansion1}
R.~Gurau, ``{The 1/N expansion of colored tensor models},''
  \href{http://dx.doi.org/10.1007/s00023-011-0101-8}{{\em Annales Henri
  Poincare} {\bfseries 12} (2011) 829--847},
\href{http://arxiv.org/abs/1011.2726}{{\ttfamily arXiv:1011.2726 [gr-qc]}}.

\bibitem{expansion2}
R.~Gurau and V.~Rivasseau, ``{The 1/N expansion of colored tensor models in
  arbitrary dimension},''
  \href{http://dx.doi.org/10.1209/0295-5075/95/50004}{{\em Europhys. Lett.}
  {\bfseries 95} (2011) 50004},
\href{http://arxiv.org/abs/1101.4182}{{\ttfamily arXiv:1101.4182 [gr-qc]}}.

\bibitem{expansion3}
R.~Gurau, ``{The complete 1/N expansion of colored tensor models in arbitrary
  dimension},'' \href{http://dx.doi.org/10.1007/s00023-011-0118-z}{{\em Annales
  Henri Poincare} {\bfseries 13} (2012) 399--423},
\href{http://arxiv.org/abs/1102.5759}{{\ttfamily arXiv:1102.5759 [gr-qc]}}.

\bibitem{BrownMap1}
J.~F. Le~Gall, ``{The topological structure of scaling limits of large planar
  maps},'' {\em Inventiones mathematicae} {\bfseries 169} no.~3, (2007)
  621--670.

\bibitem{BrownMap2}
J.~F. Le~Gall, ``{Geodesics in large planar maps and in the Brownian map},''
  {\em Acta mathematica} {\bfseries 205} no.~2, (2010) 287--360.

\bibitem{BrownMap3}
J.~F. Le~Gall {\em et~al.}, ``{Uniqueness and universality of the Brownian
  map},'' {\em The Annals of Probability} {\bfseries 41} no.~4, (2013)
  2880--2960.

\bibitem{melbp}
R.~Gurau and J.~P. Ryan, ``{Melons are branched polymers},''
  \href{http://dx.doi.org/10.1007/s00023-013-0291-3}{{\em Annales Henri
  Poincare} {\bfseries 15} no.~11, (2014) 2085--2131},
\href{http://arxiv.org/abs/1302.4386}{{\ttfamily arXiv:1302.4386 [math-ph]}}.

\bibitem{Bonzom:2015axa}
V.~Bonzom, T.~Delepouve, and V.~Rivasseau, ``{Enhancing non-melonic
  triangulations: A tensor model mixing melonic and planar maps},''
  \href{http://dx.doi.org/10.1016/j.nuclphysb.2015.04.004}{{\em Nucl.Phys.}
  {\bfseries B895} (2015) 161--191},
\href{http://arxiv.org/abs/1502.01365}{{\ttfamily arXiv:1502.01365 [math-ph]}}.

\bibitem{double}
E.~Brezin and V.~A. Kazakov, ``{Exactly solvable field theories of closed
  strings},''
\href{http://dx.doi.org/10.1016/0370-2693(90)90818-Q}{{\em Phys. Lett.}
  {\bfseries B236} (1990) 144--150}.

\bibitem{double1}
M.~R. Douglas and S.~H. Shenker, ``{Strings in less than one-dimension},''
\href{http://dx.doi.org/10.1016/0550-3213(90)90522-F}{{\em Nucl. Phys.}
  {\bfseries B335} (1990) 635}.

\bibitem{double2}
D.~J. Gross and A.~A. Migdal, ``{Nonperturbative two-dimensional quantum
  gravity},''
\href{http://dx.doi.org/10.1103/PhysRevLett.64.127}{{\em Phys. Rev. Lett.}
  {\bfseries 64} (1990) 127}.

\bibitem{DGR}
S.~Dartois, R.~Gurau, and V.~Rivasseau, ``{Double scaling in tensor models with
  a quartic interaction},''
  \href{http://dx.doi.org/10.1007/JHEP09(2013)088}{{\em JHEP} {\bfseries 1309}
  (2013) 088},
\href{http://arxiv.org/abs/1307.5281}{{\ttfamily arXiv:1307.5281 [hep-th]}}.

\bibitem{GurSch}
R.~Gurau and G.~Schaeffer, ``{Regular colored graphs of positive degree},''
  \href{http://arxiv.org/abs/1307.5279}{{\ttfamily arXiv:1307.5279 [math.CO]}}.

\bibitem{Bonzom:2014oua}
V.~Bonzom, R.~Gurau, J.~P. Ryan, and A.~Tanasa, ``{The double scaling limit of
  random tensor models},''
  \href{http://dx.doi.org/10.1007/JHEP09(2014)051}{{\em JHEP} {\bfseries 1409}
  (2014) 051},
\href{http://arxiv.org/abs/1404.7517}{{\ttfamily arXiv:1404.7517 [hep-th]}}.

\bibitem{Delepouve:2015nia}
T.~Delepouve and R.~Gurau, ``{Phase Transition in Tensor Models},''
\href{http://arxiv.org/abs/1504.05745}{{\ttfamily arXiv:1504.05745 [hep-th]}}.

\bibitem{Oriti:2011jm}
D.~Oriti, ``{The microscopic dynamics of quantum space as a group field
  theory},''
  in  {\sl Foundations of space and time}, 
G. Ellis, et al. (eds.) (Cambridge University Press, Cambridge UK, 2012),
\href{http://arxiv.org/abs/1110.5606}{{\ttfamily arXiv:1110.5606 [hep-th]}}.

\bibitem{Baratin:2011aa}
A.~Baratin and D.~Oriti, ``{Ten questions on group field theory (and their
  tentative answers)},''
  \href{http://dx.doi.org/10.1088/1742-6596/360/1/012002}{{\em J. Phys. Conf.
  Ser.} {\bfseries 360} (2012) 012002},
\href{http://arxiv.org/abs/1112.3270}{{\ttfamily arXiv:1112.3270 [gr-qc]}}.

\bibitem{Krajewski:2012aw}
T.~Krajewski,
 ``{Group field theories},''
  PoS QGQGS {\bf 2011}, 005 (2011)
\href{http://arxiv.org/abs/1210.6257}{{\ttfamily arXiv:1210.6257 [gr-qc]]}}.

\bibitem{tt2}
S.~Carrozza, D.~Oriti, and V.~Rivasseau, ``{Renormalization of an SU(2)
  tensorial group field theory in three dimensions},''
\href{http://dx.doi.org/10.1007/s00220-014-1928-x}{{\em Commun.
  Math. Phys.} {\bfseries  330} (2014) 581-637},
\href{http://arxiv.org/abs/1303.6772}{{\ttfamily arXiv:1303.6772 [hep-th]}}.

\bibitem{Rivasseau:2013uca}
V.~Rivasseau, ``{The tensor track, III},''
\href{http://arxiv.org/abs/1311.1461}{{\ttfamily arXiv:1311.1461 [hep-th]}}.

\bibitem{BenGeloun:2011rc}
J.~Ben~Geloun and V.~Rivasseau, ``{A renormalizable 4-Dimensional tensor field
  theory},'' \href{http://dx.doi.org/10.1007/s00220-012-1549-1}{{\em Commun.
  Math. Phys.} {\bfseries 318} (2013) 69--109},
\href{http://arxiv.org/abs/1111.4997}{{\ttfamily arXiv:1111.4997 [hep-th]}}.

\bibitem{Geloun:2014kpa}
J.~Ben~Geloun, ``{Renormalizable Models in Rank $d\geq 2$ Tensorial Group Field
  Theory},'' \href{http://dx.doi.org/10.1007/s00220-014-2142-6}{{\em Commun.
  Math. Phys.} {\bfseries 332} (2014) 117--188},
\href{http://arxiv.org/abs/1306.1201}{{\ttfamily arXiv:1306.1201 [hep-th]}}.

\bibitem{Samary:2014oya}
D.~O. Samary, C.~I. P\'erez-S\'anchez, F.~Vignes-Tourneret, and R.~Wulkenhaar,
  ``{Correlation functions of just renormalizable tensorial group field theory:
  The melonic approximation},''
\href{http://arxiv.org/abs/1411.7213}{{\ttfamily arXiv:1411.7213 [hep-th]}}.

\bibitem{Lahoche:2015ola}
V.~Lahoche, D.~Oriti, and V.~Rivasseau, ``{Renormalization of an Abelian Tensor
  Group Field Theory: Solution at Leading Order},''
  \href{http://dx.doi.org/10.1007/JHEP04(2015)095}{{\em JHEP} {\bfseries 1504}
  (2015) 095},
\href{http://arxiv.org/abs/1501.02086}{{\ttfamily arXiv:1501.02086 [hep-th]}}.

\bibitem{BenGeloun:2012yk}
J.~Ben~Geloun, ``{Two and four-loop $\beta$-functions of rank 4 renormalizable
  tensor field theories},''
  \href{http://dx.doi.org/10.1088/0264-9381/29/23/235011}{{\em Class. Quant.
  Grav.} {\bfseries 29} (2012) 235011},
\href{http://arxiv.org/abs/1205.5513}{{\ttfamily arXiv:1205.5513 [hep-th]}}.

\bibitem{Carrozza:2014rba}
S.~Carrozza, ``{Discrete renormalization group for SU(2) tensorial group field
  theory},''
\href{http://arxiv.org/abs/1407.4615}{{\ttfamily arXiv:1407.4615 [hep-th]}}.

\bibitem{Nguyen:2014mga}
V.~A. Nguyen, S.~Dartois, and B.~Eynard, ``{An analysis of the intermediate
  field theory of T$^4$ tensor model},''
  \href{http://dx.doi.org/10.1007/JHEP01(2015)013}{{\em JHEP} {\bfseries 1501}
  (2015) 013},
\href{http://arxiv.org/abs/1409.5751}{{\ttfamily arXiv:1409.5751 [math-ph]}}.

\bibitem{expansioin6}
R.~Gurau, ``{The 1/N Expansion of Tensor Models Beyond Perturbation Theory},''
  \href{http://dx.doi.org/10.1007/s00220-014-1907-2}{{\em Commun. Math. Phys.}
  {\bfseries 330} (2014) 973--1019},
\href{http://arxiv.org/abs/1304.2666}{{\ttfamily arXiv:1304.2666 [math-ph]}}.

\bibitem{Nachbin}
L.~Nachbin, {\em The Haar integral}.
\newblock D. Van Nostrand, Princeton, N.J., Toronto, New York, 1965.

\bibitem{petro2008moduli}
M.~Petro, {\em Moduli Spaces of Riemann Surfaces}.
\newblock University of Wisconsin--Madison, 2008.

\end{thebibliography}
\end{document}